\documentclass{llncs}
\usepackage{amssymb}
\begin{document}
\frontmatter          % for the preliminaries
\pagestyle{headings}  % switches on printing of running heads
%\addtocmark{Hamiltonian Mechanics} % additional mark in the TOC
%
%
\mainmatter              % start of the contributions
\title{Computational Complexity on Signed Numbers}
%
% abbreviated title (for running head)
% also used for the TOC unless
% \toctitle is used
\titlerunning{Computational Complexity on Signed Numbers}
\author{Stefan Jaeger}
\authorrunning{Stefan Jaeger} % abbreviated author list (for running head)
\tocauthor{Stefan Jaeger}
\institute{%Jaeger Research,\\
%54293 Trier, Germany\\
%\email{pnp@jaeger.info},\\
%WWW home page: \texttt{http://www.jaeger.info}}
%\texttt{http://www.jaeger.info}}
\texttt{www.jaeger.info}}

\maketitle              % typeset the title of the contribution

\begin{abstract}
This paper presents a new representation of natural numbers and discusses its consequences for computability and computational complexity. The paper argues that the introduction of the first Peano axiom in the traditional definition of natural numbers is not essential. It claims that natural numbers remain usable in traditional ways without assuming the existence of at least one natural number. However, the uncertainty about the existence of natural numbers translates into every computation and introduces intrinsic uncertainty that cannot be avoided. The uncertainty in the output of a computation can be reduced, though, at the expense of a longer runtime and thus higher complexity. For the new representation of natural numbers, the paper claims that, with the first Peano axiom, {\bf P} is equal to {\bf NP}, and that without the first Peano axiom, {\bf P} becomes a proper subset of {\bf NP}.
\end{abstract}

%%%%%%%%%%%%%%%%%%%%%%%%%%%%%%%%%%%%%%%%%%%%%%%%%%%%%%%%%%%%%%%%%%
% Begin of main text
%%%%%%%%%%%%%%%%%%%%%%%%%%%%%%%%%%%%%%%%%%%%%%%%%%%%%%%%%%%%%%%%%%

\section{Introduction}
Natural numbers are the heart of mathematics and related sciences. Besides from their obvious practical use as a measure for counting, natural numbers are of great theoretical importance. They form the basis for higher-level number constructs, such as rational numbers. In logic and theoretical computer science, they are instrumental in solving important problems by allowing proof techniques such as diagonalization for instance. However, despite of their apparent simplicity and paramount importance, it turns out that the definition of natural numbers is not as straightforward as one would perhaps think. The typical framework used to formalize natural numbers is set theory, and in particular the set of Peano axioms named after the 19th century Italian mathematician Giuseppe Peano. Like all formulations of the more complex Zermelo-Fraenkel set theory, the Peano axioms also assert the existence of at least one set. For Peano axioms, this set is typically the empty set, which represents the number zero. It is exactly here, where the present paper brings in its criticism. It argues that the existence of the empty set is not necessarily required to define the natural numbers. In fact, it argues that we do not need to assert the existence of a single number in order to define natural numbers. Typically, natural numbers rest inductively on a base number, such as zero. When this number exists, we can induce an infinite set of numbers by the principle of induction. The present paper argues that even when we do not know whether the base number exists, we can still compute a meaningful successor that may or may not exist. In particular, the paper proposes a relaxed version of the Peano axioms that dispenses with the first Peano axiom, which guarantees the existence of the base number. The general idea is to allow different codings for a natural number, depending on one's belief in the existence of the number, ranging from the most efficient binary code to a simple length encoding. Using the relaxed Peano axioms, it is still possible to perform computations in the traditional manner, but the representation of a natural number is not unique anymore.

The following sections present an informal discussion of this approach.
The paper structure is as follows: Section~\ref{NaturalNumbersSection} recapitulates the traditional inductive definition of natural numbers and proposes an alternative definition that avoids the first Peano axiom. Section~\ref{IntrinsicUncertaintySection} explains the ramifications of removing the first Peano axiom, discussing the intrinsic uncertainty introduced by the relaxed set of Peano axioms and its effect on computability. Section~\ref{ComputationalComplexitySection} shows the consequences for computational complexity followed by a conclusion summarizing the results.

\section{Natural Numbers}
\label{NaturalNumbersSection}

Natural numbers are usually characterized by the Peano axioms, which can be summarized as follows~\cite{peano}:
\begin{enumerate}
\item There is a natural number~$0$ (the first Peano axiom).
\item Every natural number~$a$ has a natural number successor $S(a)$.
\item There is no natural number whose successor is~$0$.
\item Distinct natural numbers have distinct successors: if $a \neq b$, then $S(a) \neq S(b)$.
\item If a property is possessed by~$0$ and also by the successor of every natural number which possesses it, then it is possessed by all natural numbers. 
\end{enumerate}
Based on the Peano axioms, it is possible to develop a set-theoretical construction of the natural numbers. For instance, a typical construction represents the number zero by the empty set and uses a successor function~$s$ to define the successor~$s(a)$ of a natural number~$a$ as $s(a) = a \cup \{a\}$ for every number~$a$. This leads to the following series of natural numbers, where each natural number is equal to the set of natural numbers smaller than it:
\begin{eqnarray} 
0 & = & \emptyset \\ \nonumber 
1 & = & \{0\} = \{\emptyset\} \\ \nonumber
2 & = & \{0,1\} = \{0, \{0\}\} = \{\emptyset, \{\emptyset\}\} \\ \nonumber
3 & = & \{0,1,2\} = \{0, \{0\}, \{0, \{0\}\}\} = \{\emptyset, \{\emptyset\}, \{\emptyset, \{\emptyset\}\}\} \\ \nonumber
& \vdots & \\ \nonumber
n & = & \{0,1,2,...,n-2,n-1\} \\ \nonumber
& = & \{0,1,2,...,n-2\} \cup \{n-1\} \\ \nonumber
& = & (n-1) \cup \{n-1\} 
\end{eqnarray}
Natural numbers generated inductively in such a way are either directly or indirectly a successor of the empty set.\bigskip

Traditionally, the first Peano axiom is taken for granted. Yet, there remains doubt about the existence of the natural numbers and the existence of the empty set in particular. The Peano axiom cannot eliminate the uncertainty about the existence of a natural number, it can merely subsume it. This means that either all the uncertainty is implicit in the first Peano axiom, and thus in the uncertainty about the existence of the root number, or, all the uncertainty is implicit in the inductively defined successors of the root number, and thus in whether a particular successor exists or does not exist. If the bit encoding of a natural number contains the bit encoding of the root number, then either the bit representing the root number or the remaining bits are uncertain.

%The next section will present an alternative construction of natural %numbers that allows uncertainty about both the existence of the root %number and about the existence of the successors of the root number at the %same time.

The present paper proposes an alternative representation of the natural numbers that does not require the empty set to exist. In fact, the uncertainty about the existence of the empty set and the existence of the natural numbers as a whole is implicit in the proposed construction of the natural numbers. The main motivation is the principle of Occam's razor, or the law of parsimony, which states that the simpler of two equally powerful theories is always superior. In the case of natural numbers, this would be the construction with the fewest assumptions. Since we do not know whether the empty set exists, and since this information is not essentially needed for the representation, the proposed definition will dispense with the first Peano axiom.

The alternative representation presented in this paper rests again inductively on the empty set. This time, however, the existence of the empty set is not required. The main idea is to add a sign bit to the classical binary representation of a natural number, indicating whether the root number zero is coded as a~$1$ (we believe in its existence), or as~$0$ (we do not believe in its existence). In the following, the text will refer to these numbers as b-numbers and use the symbol $\mathbb{B}$~for the set of b-numbers. Furthermore, let $b$ be a b-number with binary representation $b_{n}b_{n-1}\ldots b_{1}| b_{0}$. Then, the last bit $b_{0}$ will be referred to as the sign of~$b$. The sign indicates whether the coding of a zero is either~$0$ ($b_{0}=1$) or~$1$ ($b_{0}=0$). The remaining bits $b_{n}b_{n-1}\ldots b_{1}$ of~$b$ encode a natural number using a binary coding scheme.
%...either the common encoding scheme for dual numbers based on powers of two or a less efficient encoding.
%Furthermore, let us agree that~$b_{1}$ is the least significant bit in the %bit encoding of~$b$, and that the length of~$b$ is the length of its bit %representation excluding the sign, i.e. $n$.

Representations of b-numbers differ in an important way from the traditional binary representations: B-numbers can have different possible string representations, depending on our belief in the existence of the root number. The b-number value coded by a binary string is not necessarily given by a sum of powers of two, as in the traditional case, but can also be defined by a less efficient encoding, such as the length of the string. The need for a less efficient encoding is due to the intrinsic uncertainty in the existence of the root number (or empty set) and the uncertainties in its successors, which introduces inevitably uncertainty into the bit representations of natural numbers. The next section will elaborate on this intrinsic uncertainty.

\section{Intrinsic Uncertainty}
\label{IntrinsicUncertaintySection}

Natural numbers have always played an important role in theoretical computer science~\cite{turing1,turing2}. In particular, the concept of computably enumerable sets, such as languages accepted by Turing machines, is intimately connected to natural numbers. A change in the foundation of natural numbers, like the one presented in this paper, will therefore have an immediate impact on enumerable sets and computability. We will see that, under the new representation of natural numbers, all computations are uncertain.

As indicated in the motivation above, there is uncertainty involved in the coding of zero, which can be either coded as a~$0$ or as a~$1$, depending on the sign. Accordingly, at least one bit in the bit representation of~$b$ is uncertain. Let us refer to this uncertainty in the encoding of~$b$ as the uncertainty~$E(b)$ of~$b$. For instance, for $b=0$, there are two possible encodings of~$b$, namely either $0|1$ or $1|0$. In the first case, the value of the actual bit encoding~$b$ is~$0$ and the value of the sign bit indicates that a zero indeed represents a zero. The second possible encoding for $b=0$, $1|0$, represents zero as a~$1$, with the value of the sign bit indicating that a one represents in fact a zero. On the other hand, for $b=1$, the two possible encodings are $0|0$ or $1|1$. The first encoding represents a one as a zero as indicated by the sign bit, while the second encoding uses a one to represent a one, which is again indicated by the sign bit. While both encodings for $b=1$ are possible, only one of the two encodings reflects the reality, i.e. the true value of the sign bit that we do not know. Accordingly, only one of the two possible encodings for $b=0$ can reflect the true value of the sign bit. Thus, the uncertainty involved in picking the right encoding for $b=0$ or $b=1$ is one bit. The following holds for both cases: When we know the sign~$b_{0}$, we do know nothing about~$b_{1}$. Vice versa, when we know the value of~$b_{1}$, we do not know the value of the sign~$b_{0}$.

The uncertainty in the bit representation of the root number extends inductively to bit strings of larger b-numbers. The bit string representation of a b-number will contain bits based on the correct encoding of the empty set (correct bits) and bits based on the wrong encoding of the empty set (incorrect bits). Given the intrinsic uncertainty in the encoding of the root number, however, it is impossible to identify the two sets of correct and incorrect bits with absolute certainty. There will always be uncertainty involved in the discrimination between correct and incorrect bits. However, one can assume that the coding of one of the two sets is fixed (known or certain bits) and the coding of the other set is uncertain (uncertain bits) in the sense that all its bits may need to be flipped to make them consistent with the remaining bits. Thus, the bit representation of a b-number consists of two bit groups. One group contains the bits that are assumed to be correct, while the other group contains the uncertain bits. Both groups make different assumptions about the existence of the empty set, and there is uncertainty as to whether all the bits in one of the groups need to be flipped. This division into two groups has immediate consequences for the enumeration process of languages accepted by Turing machines. When using a bit representation for the enumeration that encodes the program code, the input to a Turing machine, and its computed output, then there must be uncertainty involved.
%We can distinguish between two extreme cases: When we consider one bit of %the enumerated bit string to be absolutely correct, then the uncertainty %in the remaining bits needs to be maximum. On the other hand, if we allow %maximum uncertainty in one bit of the bit string, then we can use a much %more efficient encoding for the remaining bit string.

The uncertainty in the bit representation of b-numbers does not need to be evenly distributed among the empty set and the other bits. All possible uncertainty distributions fall between two extreme cases: In the first case, the encoding of the root number is assumed to be known with no uncertainty involved at all. Hence, the remaining bits of the bit string representation of a b-number must all be random. Consequently, the only way to code a natural number in this way is to use the length of the bit string. Assuming no uncertainty in the encoding of the empty set thus comes with a cost, namely the cost of having to choose an inefficient coding of natural numbers. In the second case, numbers are encoded in the most efficient way possible, meaning that every bit is assumed to be known except for the bit encoding the root number; the sign bit, which is uncertain. The second case thus chooses the most efficient; i.e. shortest, encoding of a natural number at the expense of knowing nothing about the encoding of the root number itself. While the first case needs $n$ bits to encode a natural number~$n$, the second case only needs $\lceil\log(n)\rceil$ bits (excluding the sign bit). All other uncertainty distributions fall between these two extremes and differ in how efficiently they can encode a natural number.

In order to quantify the uncertainty~$E(b)$ of a b-number for arbitrary~$b$, let~$I(p)$ denote the standard entropy for a probabilistic event with probability~$p$~\cite{shannon}:
\begin{equation}
I(p) = -p\cdot log_{2}(p) - (1-p) \cdot log_{2}(1-p).
\label{twoCaseEntropyEquation}
\end{equation}
The following theorem then establishes a connection between~$I(p)$ and the uncertainty~$E(b)$ in the representation of a b-number~$b$.

\begin{theorem} [Entropy theorem]\label{EntropyBoundTheorem}
For any b-number~$b>0$, the uncertainty $E(b)$ is bound by $I\left(1/(b+1)\right) \le E(b) \le I\left(1/(\lceil\log_{2}(b+1)\rceil+1)\right)$.
%\label{EntropyBoundTheorem}
\end{theorem}

\begin{proof}
Let~$b$ be a b-number with $b=b_{n}b_{n-1}\ldots b_{1}| b_{0}$. For~$b=1$, we have already established the validity of Theorem~\ref{EntropyBoundTheorem}. In this case, the uncertainty in the bit representation of~$b$ is exactly one bit, as both delimiting expressions in Theorem~\ref{EntropyBoundTheorem} are $I(0.5) = 1$. The uncertainty~$E(b)$ in the representation of b-numbers greater one is delimited by two extreme cases: In the first case, all uncertainty is contained in the sign bit and the most efficient traditional binary coding is used for the actual coding, $b_{n}b_{n-1}\ldots b_{1}$, of a natural number. In this case, the uncertainty manifests in the question of whether to flip the sign $b_{0}$ so that it becomes consistent with the remaining bits in the actual coding of the natural number. This is equivalent to saying that the sign bit uses the opposite encoding of~$0$ and~$1$. If the encoding $b_{n}b_{n-1}\ldots b_{1}$ uses a~$1$ to represent a~$1$, and a~$0$ to represent a~$0$, the sign bit~$b_{0}$ uses the opposite encoding, using a~$1$ to represent a~$0$ and a~$0$ to represent a~$1$. Vice versa, if the encoding $b_{n}b_{n-1}\ldots b_{1}$ uses a~$1$ for a~$0$ and a~$0$ for a~$1$, then the coding of the sign bit~$b_{0}$ uses a~$1$ to encode a~$1$ and a~$0$ to encode a~$0$. The question is what encoding to choose, so that the encoding of the b-number becomes consistent in the sense that $b_{n}b_{n-1}\ldots b_{1}| b_{0}$ uses an identical encoding. In other words, which part of $b_{n}b_{n-1}\ldots b_{1}| b_{0}$, either $b_{n}b_{n-1}\ldots b_{1}$ or $b_{0}$ do we need to change? This is a binary decision process, with two event probabilities involved, namely $1/(\lceil\log_{2}(b+1)\rceil+1)$ and $1-1/(\lceil\log_{2}(b+1)\rceil+1)$, where $n = \lceil\log_{2}(b+1)\rceil$ is the length of the actual encoding $b_{n}b_{n-1}\ldots b_{1}$. Hence, the uncertainty $E(b)$ in the representation $b_{n}b_{n-1}\ldots b_{1}| b_{0}$ of~$b$ is $I\left(1/(\lceil\log_{2}(b+1)\rceil+1)\right)$.

The second extreme case occurs when all uncertainty is in the actual encoding of~$b$, i.e. in $b_{n}b_{n-1}\ldots b_{1}$, while the sign bit~$b_{0}$ is assumed to be certain. Here, the only way to encode~$b$ is by the length of its bit string $b_{n}b_{n-1}\ldots b_{1}$ because all bits are considered uncertain, and therefore $b=n$. Again, a decision needs to be made as to what bits to flip so that the entire representation $b_{n}b_{n-1}\ldots b_{1}| b_{0}$ becomes consistent. This is a binary decision process with probabilities $1/(b+1)$ and $1-1/(b+1)$. Thus, the uncertainty $E(b)$ in the representation of~$b$ is $I\left(1/(b+1)\right)$.
The theorem now follows directly from $I\left(1/(b+1)\right) < I\left(1/(\lceil\log_{2}(b+1)\rceil+1)\right)$.
\end{proof}

Based on the intrinsic uncertainty in b-numbers, let us now define the uncertainty~$E(C)$ of a computation~$C$. In order to do so, let the triple~$C = (T,b,o)$ represent the computation of a machine~$T$ on input~$b$ that results in an output bit~$o$. The value of~$o$ depends on whether or not~$T$ accepts its input~$b$. We can subsume the computation~$C$ into a single number by simply concatenating the traditional individual bit encodings of~$T$,$b$, and~$o$. To convert this traditional encoding into a b-number encoding, one of the bits need to be chosen as the sign bit. Without restriction of generality, we can assume that the result bit~$o$ plays the role of the sign bit. In this way, it is possible to enumerate all inputs accepted or rejected by~$T$ using b-numbers. The uncertainty~$E(C)$, or entropy, of the computation~$C$ is now defined as the uncertainty in its b-number encoding. It follows straightforwardly that all computations on b-numbers are uncertain, in the sense that their output bit is uncertain, as summarized by the following uncertainty corollary:

\begin{corollary}\label{UncertaintyCorollary}
Every computation~$C = (T,b,o)$ on b-numbers is uncertain.
%\label{UncertaintyCorollary}
\end{corollary}

\begin{proof}
According to Theorem~\ref{EntropyBoundTheorem}, every b-number~$b$ contains uncertainty, and the minimum uncertainty possible in its encoding is $I\left(1/(b+1)\right)$. Consequently, any enumerating b-number encoding of a computation~$C$ will contain uncertainty. Thus, the average uncertainty in each bit, including the uncertainty in the output bit~$o$, will be greater than zero.
\end{proof}

Corollary~\ref{UncertaintyCorollary} states that any computation based on b-numbers has an uncertainty greater than zero. It is, however, possible to reduce the uncertainty of any computation to an infinitesimal value. In order to show this, let two computations $C = (T,b,o)$ and $C' = (T',b',o')$ be equivalent if there exists a bijective mapping $M:\mathbb{B}\rightarrow\mathbb{B}$, so that $b'=M(b)$ and $o=o'$ for all inputs~$b$, where~$\mathbb{B}$ denotes the set of b-numbers. Thus, an equivalent computation produces the same output but uses a different encoding of the input. By choosing a proper re-encoding of the input, i.e. a suitable mapping~$M$, an equivalent computation can reduce the uncertainty of the computation to any infinitesimal value~\cite{jaegerArXivUncertainty}. The following entropy reduction theorem summarizes this result.

\begin{theorem} [Entropy reduction theorem]\label{EntropyReductionTheorem}
Let $C = (T,b,o)$ be a computation on b-numbers. For any $\epsilon>0$, there exists an equivalent computation $C' = (T',M(b),o)$ with $E(C')< \epsilon$ for a bijective mapping $M:\mathbb{B}\rightarrow\mathbb{B}$.
%\label{EntropyReductionTheorem}
\end{theorem}

\begin{proof}
To prove Theorem~\ref{EntropyReductionTheorem} let us construct an equivalent computation~$C'$ that satisfies the required uncertainty constraint. The upper uncertainty threshold~$\epsilon$ in Theorem~\ref{EntropyReductionTheorem} can be written in the following form for a probability~$p^\ast$:
\begin{equation}
\epsilon = -p^\ast \cdot log_{2}(p^\ast) - (1-p^\ast) \cdot log_{2}(1-p^\ast)
\label{EntropyReductionProofEquation}
\end{equation}
Now, let~$M$ be a machine that implements the bijective mapping between b-numbers. It can do so by augmenting the original input string~$b$ with additional bits whose values are based on the same encoding of the empty set. When~$M$ keeps adding these bits until the ratio of the two bit groups used in the encoding of~$C$ becomes smaller than~$p^\ast$, the uncertainty of $C$ will eventually become smaller than~$\epsilon$. Once the required uncertainty is reached, a simple removal of the added bits followed by the execution of the original computation $T(b)$ will produce the same result as~$C$. Hence, $C' = (T',M(b),o)$ with $T'=T\circ M^{-1}$ is an equivalent computation to $C = (T,b,o)$ that satisfies the required uncertainty constraints $E(C')< \epsilon$.
\end{proof}

\section{Computational Complexity}
\label{ComputationalComplexitySection}

This section will discuss the consequences that b-numbers have for computational complexity~\cite{papadimitriou2003computational}. It shows how b-numbers affect the relationship between the complexity classes {\bf P} and {\bf NP}. In particular, it shows how the first Peano axiom is intimately connected with the hierarchy of complexity classes when we use b-number encodings. The results of Theorem~\ref{EntropyBoundTheorem} and Theorem~\ref{EntropyReductionTheorem} will play an important role here. The next paragraph is a brief summary of the main definitions~\cite{cormen,garey-book}.

Let $\Sigma$ be a finite input alphabet associated with a Turing machine~$M$, with $|\Sigma|\ge 2$. Furthermore, let $\Sigma^*$ be the set of finite strings over $\Sigma$. Then a language over $\Sigma$ is a subset L of $\Sigma^*$. The Turing machine~$M$ is said to accept an input string $w\in\Sigma^*$ if the computation terminates in an accepting state. On the other hand, $M$ does not accept~$w$ if the computation either terminates in a rejecting state or if the computation fails to terminate. The language~$L(M)$ accepted by~$M$ is defined by
\begin{equation}
L(M) = \{w\in\Sigma^* \, | \, M \mbox{accepts}~w\}
\end{equation}
Now, let $t_M(w)$ be the number of steps in the computation of~$M$ on input~$w$.
If the computation for input~$w$ never halts, then $t_M(w) = \infty$. Using the run time expression $t_M(w)$, we can define the complexity class~{\bf P} based on the worst case run time $T_M(n)$ of~$M$ for all $n\in\mathbb{N}$:
\begin{equation}
T_M(n) = \max\{t_M(w) \, | \, w\in\Sigma^n\}
\end{equation}
where $\Sigma^n$ is the set of all strings over $\Sigma$ of length~$n$. A Turing machine~$M$ is said to run in polynomial time if there exists a~$k$ such that for all~$n$, $T_M(n) \le n^k+k$. Then, the complexity class~{\bf P} is the set of all languages for which there exists a polynomial-time Turing machine~$M$, i.e.:
\begin{eqnarray}
{\bf P} & = \{ & L \, | \, L=L(M) \, \mbox{for some Turing machine~M} \\
& & \mbox{which runs in polynomial time}\} \nonumber
\end{eqnarray}
If a problem can be solved in polynomial time on a non-deterministic Turing machine, we say the Problem is in {\bf NP}. Obviously, {\bf P} is a subset of {\bf NP}, but whether it is a proper subset or not is an open question.
Without restriction of generality, we can confine $\Sigma$ to a binary alphabet; i.e. $\Sigma = \{0,1\}$. For each member $L'$ of {\bf P}, we can then count the words in $L'$ using b-numbers. One possible enumeration is to use the binary representation for each word in~$\Sigma$ and add a bit that indicates whether or not~$M$ accepts the word. Here, it shows that uncertainty in the representation of a b-number has an immediate effect on the enumeration of~$L'$. There is either uncertainty in the input of~$M$ or the output of~$M$, or both.

The proposed representation of natural numbers and the uncertainty it introduces into enumerations affects the computational complexity of problems. In particular, b-numbers add a new facet to problems by allowing to specify upper bounds on the uncertainty in the output of Turing machines.
As a first result, the following Theorem~\ref{PeanoTheorem} establishes a direct connection between the Peano axiom and the relationship between {\bf P} and {\bf NP}.

\begin{theorem} [P theorem]\label{PeanoTheorem}
For computations on b-numbers, {\bf P}={\bf NP} follows from the first Peano axiom.
%\label{PeanoTheorem}
\end{theorem}

\begin{proof}
The first Peano axiom introduces a crisp statement about the coding of the empty set. In other words, the first Peano axiom assumes a specific value of the sign bit in the b-number representation of a natural number. Obviously, as an axiom cannot be proven, this leaves two extreme possibilities for the uncertainty in a computation~$C$. We can either base our trust on the first Peano axiom, which means that all uncertainty lies in the remaining bit encoding of~$C$, or we can accept that all uncertainty lies in the first Peano axiom, in which case all remaining bits are certain.
These are the two extreme cases of Theorem~\ref{EntropyBoundTheorem}, which means that the uncertainty~$E(C)$ of~$C$ is either $I\left(1/(n+1)\right)$ or $I\left(1/(\lceil\log_{2}(n+1)\rceil+1)\right)$, where $n$ is the number of the traditional encoding of $(T,b)$ for the computation $C = (T,b,o)$. Now, in order to solve the problem for input~$b$, $T$ does not need to solve the problem in its entirety. It just needs to run until the required uncertainty is reached, after which it can output any result bit. In fact, no actual computation is needed for $T$~to solve the problem for input~$b$. It suffices, and is essential at the same time, to add bits until the required uncertainty is reached, using a mapping~$M$ like the one used in Theorem~\ref{EntropyReductionTheorem}. The ratio of the number of certain bits to the number of uncertain bits is either $1/(n+1)$ or $1/(\lceil\log_{2}(n+1)\rceil+1)$. Hence, mapping~$M$ needs to add either an exponential or a polynomial number of bits. Any implementation of~$M$ will therefore take either an exponential or a polynomial number of processing steps for any computation. Consequently, all computations on b-numbers take either exponential or polynomial runtime. Thus, by assuming that the input bits are correct, we can ensure a polynomial runtime, and hence {\bf P}={\bf NP}.
\end{proof}

Theorem~\ref{PeanoTheorem} shows that the complexity classes {\bf P} and {\bf NP} become equal once the first Peano axiom is introduced under b-numbers. According to Section~\ref{NaturalNumbersSection}, however, the first Peano axiom is not needed for the construction of b-numbers. In fact, the next theorem shows that the relationship between {\bf P} and {\bf NP} changes once the first Peano axiom is excluded from the set of axioms. In this case, {\bf P} becomes a proper subset of {\bf NP}. In order to prove this, we can take advantage of graded uncertainty values, which are now possible due to the absence of the first Peano axiom.
The idea is to define a function that is in {\bf NP} but not in {\bf P}, which means that {\bf P} $\neq$ {\bf NP}, and thus {\bf P} $\subset$ {\bf NP}. Theorem~\ref{PNPTheorem} provides such a function.

\begin{theorem} [PNP theorem]\label{PNPTheorem}
Given b-number encodings, let~$T(T',b)$ be a machine that checks with a maximum uncertainty of $I(1/(2^b+1))$ whether $T'$ accepts input~$b$. Then, $T$ is in {\bf NP} but not in {\bf P}.
%\label{PNPTheorem}
\end{theorem}

\begin{proof}
Theorem~\ref{PNPTheorem} requires that the uncertainty $E(C)$~of any computation $C = (T,T'b,o)$ is lower than $I(1/(2^b+1))$. According to the entropy reduction theorem (Theorem~\ref{EntropyReductionTheorem}), this uncertainty can be reached by adding additional bits to the input in combination with an appropriate mapping~$M$ like in the proof of Theorem~\ref{EntropyReductionTheorem}. For instance, one could use the program code of~$T$ for the bits that need to be added. Obviously, an exponential number of processing steps with $\mathcal O(2^n)$ is sufficient to implement mapping~$M$ and reach the required uncertainty. Thus, $T$ is indeed a member of~{\bf NP}. Furthermore, it is not possible to reach the required upper uncertainty threshold with a polynomial number of processing steps in the worst case scenario. To show this, one can use a diagonal argument. For instance, in the case of $T'=T$ and $b=T$, the computation $C = (T,TT,o)$ takes an exponential number of processing steps to meet the uncertainty requirement. This becomes evident when we look at the encoding of the string~$TTT$. In order for the ratio of the number of certain bits to the number of uncertain bits to be smaller than or equal to $1/2^b$, $T$~has to execute at least an exponential number of processing steps. Thus~$T$ cannot be in {\bf P} and {\bf P} must be a proper subset of {\bf NP}.
\end{proof}
%On b-numbers, it is possible to define a problem that is in NP but not in %P, and hence P $\neq$ NP.

\section{Conclusion}

This paper has shown that the representation of natural numbers affects the computability and computational complexity of programs. The proposed representation based on a set of relaxed Peano axioms, which dispenses with the first Peano axiom, introduces intrinsic uncertainty into the representation of natural numbers. Under the traditional definition of natural numbers, the existence of the empty set is guaranteed by the first Peano axiom, which is either certain or uncertain. This important fact has been underestimated in the literature so far. The proposed representation without the first Peano axiom allows a graded belief in the existence of the empty set. It allows multiple possible representations for a natural number, which are all inductively based on the empty set, but differ in the uncertainty they impose on it.

Naturally, the intrinsic uncertainty in the representation of b-numbers affects the computational complexity of problems. The paper has argued that, for b-number encodings, the computational complexity of a problem depends on the uncertainty that we allow in the output of the Turing machine solving the problem. The first Peano axiom makes a crisp binary decision on the existence of a natural number. Since the statement made by the first Peano axiom can be coded in one bit, which is either certain or uncertain, we are facing two extreme cases. Consequently, after the introduction of the first Peano axiom, a graded belief in the existence of the empty set is no longer possible, and the paper claims that {\bf P} becomes equal to {\bf NP} for this case. If we allow arbitrary uncertainty in the output by removing the first Peano axiom, however, then {\bf P} becomes a proper subset of {\bf NP}, as there exists a function in {\bf NP} that needs exponential runtime to reach the required uncertainty in its output. These results may help shed light on the relationship between the complexity classes {\bf P}~and {\bf NP}.

%%%%%%%%%%%%%%%%%%%%%%%%%%%%%%%%%%%%%%%%%%%%%%%%%%%%%%%%%%%%%%%%%%
% End of main text
%%%%%%%%%%%%%%%%%%%%%%%%%%%%%%%%%%%%%%%%%%%%%%%%%%%%%%%%%%%%%%%%%%

%\bibliography{LITERATUR}
\bibliography{JaegerICALP2011}
\bibliographystyle{splncs03}

\end{document}